\documentclass{sig-alternate-10pt}

\usepackage{IEEEtrantools}
\usepackage{graphicx}
\graphicspath{{./figs/}}
\usepackage{mathrsfs}
\usepackage{mathabx}
\usepackage{caption}

\usepackage{lastpage}

\newcommand{\nop}[1]{}

\newtheorem{example}{Example}
\newtheorem{theorem}{Theorem}
\newtheorem{corollary}{Corollary}
\newtheorem{lemma}{Lemma}

\newtheorem{remark}{Remark}

\begin{document}

\title{Towards a Calculus for Wireless Networks}

\numberofauthors{2}
\author{
\alignauthor
	Fengyou Sun\\
	\affaddr{IIK, NTNU,\\ Department of Information Security and Communication Technology,\\ Norwegian University of Science and Technology}\\
	\email{sunfengyou@gmail.com}\\	
\alignauthor
	Yuming Jiang\\
	\affaddr{IIK, NTNU,\\ Department of Information Security and Communication Technology,\\ Norwegian University of Science and Technology}\\
	\email{ymjiang@ieee.org}\\	
\and  
}


\maketitle

\begin{abstract}

This paper presents a set of new results directly exploring the special characteristics of the wireless channel capacity process. An appealing finding is that, for typical fading channels, their instantaneous capacity and cumulative capacity are both light-tailed. A direct implication of this finding is that the cumulative capacity and subsequently the delay and backlog performance can be upper-bounded by some exponential distributions, which is often assumed but not justified in the wireless network performance analysis literature. In addition, various bounds are derived for distributions of the cumulative capacity and the delay-constrained capacity, considering three representative dependence structures in the capacity process, namely comonotonicity, independence, and Markovian. To help gain insights in the performance of a wireless channel whose capacity process may be too complex or detailed information is lacking, stochastic orders are introduced to the capacity process, based on which, results to compare the delay and delay-constrained capacity performance are obtained. Moreover, the impact of self-interference in communication, which is an open problem in stochastic network calculus (SNC), is investigated and original results are derived. The obtained results in this paper complement the SNC literature, easing its application to wireless networks and its extension towards a calculus for wireless networks. 

\end{abstract}



\section{Introduction}
\label{sec-intr}

In wireless communication, there will be a continuing wireless data explosion and an increasing demand on higher data rate and less latency. It has been depicted that the amount of IP data handled by wireless networks will exceed $500$ exabytes by 2020, the aggregate data rate and edge rate will increase respectively by $1000\times$ and $100\times$ from 4G to 5G, and the round-trip latency needs to be less than $1$ms in 5G \cite{andrews2014what}.
Evidently, it becomes more and more crucial to explore the ultimate capacity that a wireless channel can provide under stringent delay constraints and to analyze what delay limit may be achieved in specific wireless channel situations. 

Information theory provides a framework for studying the performance limits in communication and the most basic measure of performance is channel capacity, i.e., the maximum rate of communication for which arbitrarily small error probability can be achieved \cite{tse2005fundamentals}. To date, wireless channel capacity has mostly been analyzed for its average rate in the asymptotic regime, i.e., ergodic capacity, or at one time instant/short time slot, i.e., instantaneous capacity. For instance, the first and second order statistical properties of instantaneous capacity have been extensively investigated, e.g., in \cite{rafiq2011statistical,renzo2010channel}.  However, such properties of wireless channel capacity are ordinarily not sufficient for use in assessing if data transmission over the channel meets its quality of service (QoS) requirements. 

Additionally, the accumulation of capacity in finite time regime, namely cumulative capacity, should be investigated. The cumulative capacity through a period is essentially the amount of data transmission service that the wireless channel provides (if there is data for transmission) or is capable of providing (if there is no data for transmission) in this period. This concept is closely related to the cumulative service used in the stochastic network calculus (SNC) theory \cite{jiang2008stochastic}, which has been used for delay and backlog analysis in wireless networks, e.g., \cite{jiang2005analysis, fidler2006end, jiang2008stochastic, jiang2010note, mahmood2011delay, zheng2013stochastic, ciucu2013towards, ciucu2014sharp, poloczek2015service, lei2016stochastic,al2016network}.
However, most of the results focus on deriving delay and backlog by mapping the physical layer model to a higher layer abstraction, e.g., Gilbert-Elliott channel model \cite{jiang2005analysis,fidler2006wlc15} and finite state Markov channel \cite{zheng2013stochastic,lei2016stochastic}, and then making use of the SNC theory to perform the analysis, without digging deep into the fundamental properties themselves of the cumulative capacity process of the wireless channel. 

The wireless channel is an open medium, susceptible to noise, interference, and other channel impediments that change randomly over time as a result of user movement and environment dynamics \cite{goldsmith2005wireless}, e.g., large scale fading due to path loss of signal as a function of distance and shadowing by obstacles between the transmitter and the receiver, and small scale fading due to constructive and destructive addition of multipath signal components. Such natures of wireless channel give special characteristics of the channel capacity process. The two (possibly most well-known)  and {\em fundamental}  ones are: the wireless channel capacity is generally a logarithm function of the fading effect, and  on the channel, if two transmissions happen at the same time, they will collide or interfere with each other. 


The objective of this paper is to investigate such characteristics and properties of the wireless channel capacity, which have little been focused in the literature, and study their implications or impacts on the delay performance of the wireless channel. They include {\em tail behavior, distribution bounds, stochastic ordering, and self-interference in communication}. 


The specific contributions of this work are several-fold as summarized in the following.
\begin{enumerate}
\item
The tail behavior of the wireless channel capacity is studied. Specifically, we prove that the distribution of the instantaneous capacity is light-tailed for a number of typical fading channels, e.g., Rayleigh, Rice, Nakagami-$m$, Weibull, and lognormal fading channels. In addition, we prove that, conditioned on the instantaneous capacity being light-tailed, the cumulative capacity is also light-tailed. (Section \ref{tail})
\item
Distribution bounds are derived for the cumulative capacity process, under different dependence structures, namely comonotonic process, additive process, and Markov additive process. For each case, the distribution of the transient capacity is also investigated. (Section \ref{bounds}) 
\item
Bounds on the tail probabilities of delay under fluid arrival, which is the best delay performance that can be guaranteed for a specific channel,  are derived for each dependence structure. Based on these, bounds on the delay constrained capacity, which is the maximum date rate that the channel can support without violating the required delay constraint, are correspondingly obtained. (Section \ref{bounds})
\item
Stochastic ordering properties of the cumulative capacity process are exploited to compare the delay performance, based on which a comparison of the delay constrained capacity performance is also obtained. (Section \ref{comparisons})
\item
The impact of the transmission self-interference characteristic on the delay performance is investigated. Results for both a single hop case and a multi-hop case are derived. (Section \ref{interference}) 
\end{enumerate}

The remainder of this paper is structured as follows. 
In Sec. \ref{basic}, basic wireless channel capacity concepts, including instantaneous capacity, cumulative capacity, transient capacity and delay-constrained capacity are introduced. Also in Sec. \ref{basic}, the necessity of studying cumulative capacity is elaborated through an analysis on the delay and backlog performance. In Sec. \ref{tail}, the tail behavior of wireless channel capacity is focused. Distribution bounds on the cumulative capacity and delay are derived in Sec. \ref{bounds}, based on which the maximal delay-constrained capacity is also presented for each considered case. Comparison approaches for different channel ordering characteristics are provided in Sec. \ref{comparisons}. The impact of transmission interferences in wireless channel is analyzed in Sec. \ref{interference}. In Sec. \ref{related}, discussion on our findings and contributions and related works is provided. Finally, the paper is concluded in Sec. \ref{conclusion}.

\section{Basic Concepts}\label{basic}

In this section, basic wireless channel capacity concepts, including instantaneous capacity, cumulative capacity, transient capacity and delay-constrained capacity are first introduced, followed by an elaboration on the necessity of studying cumulative capacity for QoS performance of a wireless channel.

\subsection{Instantaneous Capacity}

The {\em instantaneous capacity} of a (wireless) communication channel at time $t$, denoted as $C(t)$ throughout this paper, is defined as the maximum mutual information over input distribution at $t$:\cite{costa2010multiple}: 
\begin{equation}\label{eq-1}
C(t) = \max\limits_{P(x)}I(X;Y|h(t))
\end{equation}
where $h(t)$ is a stochastic process describing the fading behavior of the channel, $\mathcal{X}$ and $\mathcal{Y}$ are respectively the input and output alphabets of the channel. 

Consider a discrete-time flat fading channel with input $x(t)$, output $y(t)$, and stationary fading process $h(t)$, which has the following complex baseband representation
\begin{equation}
y(t) = h(t)x(t) + n(t),
\end{equation}
where $n(t)$ is an i.i.d. white Gaussian noise process $\mathcal{C}\mathcal{N}(0,N_0)$. Then,  conditional on a realization of $h(t)$, the mutual information 
can be expressed as \cite{telatar1999capacity} 
\begin{equation}\label{eq-2}
I(X;Y|h(t)) = \sum\limits_{x\in\mathcal{X},y\in\mathcal{Y}} P(x,y|h_t)\log_2\frac{P(x,y|h_t)}{P(x|h_t)P(y|h_t)}.
\end{equation}

{
Particularly, for a single input single output channel, 
if the channel side information is only known at the receiver,
the instantaneous capacity can be found by solving the right hand side of (\ref{eq-1}) with (\ref{eq-2}) as \cite{tse2005fundamentals}
\begin{equation}\label{def-ic}
C(t) = W\log_{2}(1+\gamma|h(t)|^{2}),
\end{equation}
where $|h(t)|$ denotes the envelope of $h(t)$, $\gamma \equiv {P}/{N_{0}W}$ denotes the average received SNR per complex degree of freedom, $P$ is the average transmission power per complex symbol, $N_{0}/2$ is the power spectral density of AWGN, and $W$ is the channel bandwidth. 
For multiple input and multiple output channels, a generalized form of (\ref{def-ic}) is available \cite{telatar1999capacity,foschini1998limits}.
}

Averaging the instantaneous capacity over the probability space of channel gain gives the {\em ergodic capacity}  that formally is defined as  \cite{telatar1999capacity}:
\begin{equation}
\overline{C} = E[C(t)]. 
\end{equation} 
The ergodic capacity of a channel is constant. It defines the maximum transmission rate of the channel with asymptotically small error probability for the code with sufficiently long length that the received codewords is affected by all fading states \cite{goldsmith2005wireless}. As implied by its definition, the ergodic capacity is a concept for infinite code length in infinite time regime.

\subsection{Cumulative Capacity}

To account for finite time regimes, the {cumulative capacity} over a time period $(s,t]$, denoted as $S(s,t)$, is defined as the sum of instantaneous capacity in this period:
\begin{equation}
S(s,t) \equiv \sum\limits_{i=s+1}\limits^{t}{C(i)}.
\end{equation} 
For $S(0,t)$, we also use $S(t) (\equiv S(0,t))$ to simplify the expression. 

The time average of the cumulative capacity through $(0,t]$ is defined as the {\em transient capacity} \cite{tse2005fundamentals}, i.e.,
\begin{equation}
\overline{C}(t) = \frac{S(t)}{t}. 
\end{equation}

Note that the transient capacity is random, which essentially defines the achievable capacity for a code with finite length that the received codewords only experience partial fading states. 
The probabilistic average of the transient capacity is expressed as
\begin{equation}
E\left[\overline{C}(t)\right] = \overline{C}
\end{equation}
where $\overline{C}$ is the ergodic capacity of the channel.  
According to the law of large numbers, the transient capacity converges to the ergodic capacity when time goes to infinity, i.e.,
\begin{equation}
\lim\limits_{t\rightarrow{\infty}}\overline{C}(t) \rightarrow \overline{C},
\end{equation}
for independent and identically distributed instantaneous capacity. 
However, the dependence in capacity may be unknown, a more general result for the transient capacity in finite time horizon is expressed by the Chebyshev inequality \cite{papoulis2002probability}
\begin{equation}
P\{|\overline{C}(t) - \overline{C} |\ge{x}\}\le{\frac{\mathrm{Var}[\overline{C}(t) ]}{x^2}},
\end{equation}
which is a basic result of concentration \cite{boucheron2013concentration}.
It indicates that more statistical properties of the cumulative process should be taken into account besides the instantaneous capacity in view of temporal behavior.

\subsection{Delay-Constrained Capacity}

In the above-introduced wireless channel capacity concepts, the delay performance of the channel in transmitting data is not touched. As highlighted in the previous section, the delay performance is also crucial. 

In words, the {\em delay-constrained capacity} is defined as the maximum data traffic rate that the channel can support without violating a desired delay constraint for transmitting the data traffic over the channel \cite{gao2012}. 

More formally, the delay-constrained capacity or throughput is defined as the maximum traffic rate that the system can support without dropping, for which the delay constraint is met \cite{gao2012}:
\begin{equation}
\overline{C}_{(d,\epsilon)} = \sup \{ \lambda: P(D(t) > d) \le \epsilon, \forall t \}
\end{equation}
where $\lambda$ denotes the average traffic rate, the delay constraint is represented by a desired delay $d$ and the allowed violation probability $\epsilon$ of this delay, and the delay $D(t)$ at time $t$ is defined as \cite{ciucu2014sharp}:
\begin{equation}
D(t) = \inf\{ d\ge{0}: A(t-d)\le A^\ast(t) \},
\end{equation}
where $A(t)$ denotes the total amount of traffic in $(0,t]$.  

\subsection{Necessity of Studying Cumulative Capacity} 

In order to find the delay-constrained capacity of the wireless channel, we first need to analyze $P(D(t) > d)$ under input $A(t)$. Let $a(t)$ denote the traffic input to the channel at time $t$. 

The wireless channel is essentially a queueing system with cumulative service process $S(t)$ and cumulative arrival process $ A(0,t)=\sum\limits_{s=0}^{t}a(s)$, 
and the temporal increment in the system is expressed as
\begin{equation}
X(t) = a(t)-C(t).
\end{equation}

The backlog in the system is a reflected process of the temporal increment $X(t)$, as \cite{asmussen2003applied} 
\begin{equation}
B(t+1) = \left[ B(t) + X(t)\right]^{+}.
\end{equation}
Throughout this paper, $B(0)=0$ is supposed, with which the backlog function can be further expressed as
\begin{equation}\label{eq-bl}
B(t) = \sup_{0\le{s}\le{t}}({A}(s,t)-{S}(s,t)). 
\end{equation}

For a system without loss, the output, denoted as $A^{\ast}(t)$, is the difference between the input and backlog, i.e., 
\begin{equation}\label{eq-ior}
A^{\ast}(t) = A(t) - B(t) = A\otimes S(t),
\end{equation}
where $f\otimes g(t) = \inf_{0\le{s}\le{t}}\{f(s)+g(s,t)\}$ is the min-plus convolution \cite{baccelli1992synchronization}.

In order to avoid nontrivial considerations, we focus in this paper on the {\em maximal} delay-constrained capacity under all types of inputs that have the same average traffic rate $\lambda$. It can be shown, based on queue comparison results \cite{muller2002comparison}, that such {\em maximal} delay-constrained capacity is achieved when $A(t)=\lambda{t}$, i.e., the input is a constant fluid input. 

In this case, it is trivial from (\ref{eq-bl}) that for backlog, there holds
\begin{equation}\label{dis-bl}
P(B(t) > x) = P\left\{ \sup_{0\le{s}\le{t}}\{\lambda(t-s)-{S}(s,t)\} >x \right\}.
\end{equation}
For delay, with its definition and (\ref{eq-bl}), the following can also be shown
\begin{equation}
P(D(t) > {d}) = P\left\{ \sup_{0\le{s}\le{t}}\{ \lambda(t-s) - S(s,t) \}\ge\lambda{d} \right\}. \label{dis-dd}
\end{equation}
With (\ref{dis-bl}) and (\ref{dis-dd}), the following relationship can be easily verified:
\begin{equation}
P(B(t) > {x}) = P\left(D(t) > {\frac{x}{\lambda}}\right). 
\end{equation}

Eq. (\ref{dis-dd}) and Eq. (\ref{dis-bl}) immediately imply that for delay and backlog performance analysis of the wireless channel, knowing its ergodic capacity only is not enough, $C(t)$ is not sufficient either as it ignores the potential dependence behavior between $C(s+1), \cdots, C(t)$ in $S(s,t)$, and it is necessary to study the stochastic behavior of the cumulative capacity process $S(t)$.

\section{Light Tail Behavior}\label{tail}

A distribution is said to be light-tailed if the tail $\overline{F}(x)=1-F(x)$ is exponentially bounded, i.e., 
\begin{equation}
\overline{F}(x)=O(e^{-\theta{x}}), 
\end{equation}
for some $\theta>0$, or equivalently, the moment generating function $\widehat{F}[\theta]$ is finite for some $\theta>0$. Otherwise, the distribution is said to be heavy-tailed \cite{asmussen2010ruin,rolski1998stochastic}.

\nop{
Particularly, for a single input single output channel, 
if the channel side information is only known at the receiver,
the instantaneous capacity can be found by solving the right hand side of (\ref{eq-1}) with (\ref{eq-2}) as \cite{tse2005fundamentals}
\begin{equation}\label{def-ic}
C(t) = W\log_{2}(1+\gamma|h(t)|^{2}),
\end{equation}
where $|h(t)|$ denotes the envelope of $h(t)$, $\gamma \equiv {P}/{N_{0}W}$ denotes the average received SNR per complex degree of freedom, $P$ is the average transmission power per complex symbol, $N_{0}/2$ is the power spectral density of AWGN, and $W$ is the channel bandwidth. 
For multiple input and multiple output channels, a generalized form of (\ref{def-ic}) is available \cite{telatar1999capacity,foschini1998limits}.
}

\begin{theorem}
\label{th-lt}
For flat fading, where the instantaneous capacity is expressed in the logarithm transform of the instantaneous channel gain, i.e., $C(t) = W\log_{2}(1+\gamma{H(t)^2})$  for all $t$, if the distribution of the fading process is not heavier than fat tail, then the distribution of the instantaneous capacity is light-tailed.
\end{theorem}

\begin{proof}
Due to stationary assumption, $H(t)$ does not change its statistical properties over time, or $H(s) \stackrel{d}{=} H(t)$ for all $s$ and $t$. Then for any $t$, we can remove the time index $t$ and write
\begin{equation}
C = W\log_{2}(1+\gamma{H^2}).
\end{equation}
Correspondingly, the tail of the instantaneous capacity is a function of the tail of the channel gain, i.e.,
\begin{equation}
\overline{F}_C(x) = \overline{F}_H\left(\sqrt{\frac{2^{\frac{x}{W}}-1}{\gamma}}\right).
\end{equation}

Let $r = \sqrt{\frac{2^{\frac{x}{W}}-1}{\gamma}}$. The tail behavior of the instantaneous capacity can be expressed in terms of that of the fading process. Specifically, for some $\theta>0$, $\overline{F}_C(x)=O(e^{-\theta{x}})$ entails
\begin{eqnarray}
\overline{F}_H(r) 
= O\left( r^{-\theta} \right),
\end{eqnarray}
which completes the proof.
\end{proof}

\begin{corollary}
If a wireless channel is Rayleigh, Rice, Nakagami-$m$, Weibull, or lognormal fading channel, its instantaneous capacity is light-tailed. 
\end{corollary}

\begin{proof}
For Weibull fading channel, the tail of fading is expressed as
\begin{equation}
\overline{F}_{H}(r) = e^{-cr^{k}},
\end{equation}
where $c(>0)$ and $k(>0)$ are some constants. Applying Taylor's theorem to expend $e^{cr^{k}}$, it is easily shown that, for some $\theta$ satisfying $k > \theta>0$ 
\begin{equation}\label{eq-te}
\lim_{r\rightarrow\infty}\frac{e^{-cr^k}}{r^{-\theta}} = \lim_{r\rightarrow\infty}\frac{r^\theta}{1+cr^k+\ldots} = 0.
\end{equation}

The limit (\ref{eq-te}) shows that though the Weibull distribution is heavy-tailed for $0<k<1$, it is lighter than the fat tail. Hence from Theorem \ref{th-lt}, the instantaneous capacity under Weibull fading is light-tailed. 

Rayleigh fading is a special case of Weibull fading with $k=2$. The distribution of its instantaneous capacity is expressed as \cite{hogstad2007exact}
\begin{equation}
F(x) = 1- e^{\frac{1-2^{\frac{x}{W}}}{\gamma}}.
\end{equation}
It is trivial to show that the tail is exponentially bounded 
\begin{equation}
\overline{F}(x) \le e^{\frac{1}{\gamma}}e^{-\theta{x}},
\end{equation}
for $0<\theta\le\frac{1}{W\gamma}2^{\frac{1}{\log{2}}}\log{2}$. Hence, the instantaneous capacity under Rayleigh fading is also light-tailed.

For Rice fading channel, the tail of the instantaneous capacity is expressed as \cite{rafiq2009statisticalcomb}
\begin{equation}
\overline{F}(x) = Q_{1}\left( \frac{s}{\sigma_0}, \frac{\sqrt{2^{x/W}-1/{\gamma_{s}{s}}}}{{\sigma_{0}}^2} \right),
\end{equation}
where $W$ is the bandwidth, $s$ the amplitude of the LOS (light of sight) component, ${\sigma_{0}}^2$ the variance of the underlying Gaussian process, and $\gamma_s$ the average SNR. According to the exponential bound of the Marcum Q-function \cite{simon2000exponential}, 
\begin{IEEEeqnarray}{rCl}
\alpha_{F} &=& {\lim\sup}_{x\rightarrow\infty}\frac{-\log\overline{F}(x)}{x} \\
&\ge& {\lim\sup}_{x\rightarrow\infty}\frac{1}{2x}{\left( \frac{\sqrt{2^{x/W}-1/{\gamma_{s}{s}}}}{{\sigma_{0}}^2} - \frac{s}{\sigma_0} \right)^2} \\
&=& \infty, \IEEEeqnarraynumspace
\end{IEEEeqnarray}
which means that the instantaneous capacity of a Rice fading channel is light-tailed \cite{rolski1998stochastic}. 

For Nakagami-$m$ fading channel \cite{rafiq2010influence}, since the square of the Nakagami-$m$ random variable follows a gamma distribution, which is light-tailed \cite{asmussen2010ruin}, the distribution of its instantaneous capacity is thus also light-tailed.

For lognormal fading channel \cite{rafiq2008influence}, since the lognormal distribution have all the moments, which means that it has a lighter tail than the fat-tailed distribution \cite{halliwell2013classifying}, thus the distribution of its instantaneous capacity is light-tailed.
\end{proof}



\begin{corollary}
For frequency-selective fading modeled by parallel $L$ independent channels, whose instantaneous capacity is expressed as $C = \sum_{\ell=1}^{L} W_{\ell}\log_{2}(1+\gamma{H_{\ell}^2})$, if the distribution of the instantaneous capacity of each sub-channel, which is $C_\ell = W_{\ell}\log_{2}(1+\gamma{H_{\ell}^2})$,  is light-tailed, so is the instantaneous capacity of the frequency-selective fading channel.
\end{corollary}

\begin{proof}
For this frequency-selective fading channel, its instantaneous capacity is by definition related to the instantaneous capacity of each sub-channel as 
\begin{equation}
C = \sum_{\ell=1}^{L} C_\ell.
\end{equation}
The tail of the distribution of the instantaneous capacity can then be expressed by \cite{jiang2008stochastic}
\begin{eqnarray}
\overline{F}_{C}(x) &=& 1 - {F}_{C_1}\circledast\ldots\circledast{F}_{C_L}(x) \nonumber \\
&\le& \overline{F}_{C_1}\otimes\ldots\otimes\overline{F}_{C_L}(x),
\end{eqnarray}
where $f\circledast{g}(x)=\int_{-\infty}^{\infty} f(x-y)dg(y)$ is the Stieltjes convolution and $f\otimes g(t) = \inf_{0\le{s}\le{t}}\{f(s)+g(s,t)\}$ is the min-plus convolution \cite{baccelli1992synchronization}. The first step is a property of sum of independent random variables, and the second step is from that the independent case is upper bounded by the distribution bound for such a sum where no dependence information is utilized  \cite{jiang2008stochastic}. The latter, as to be illustrated in the proof of the next theorem, is light-tailed. 
\end{proof}

\begin{theorem}
Consider a wireless channel. If the distribution of its instantaneous capacity at any time is light-tailed, the distribution of the cumulative capacity is also light-tailed. In addition, the distribution of the cumulative capacity of a concatenation of such wireless channels is light-tailed as well.
\end{theorem}

\begin{proof}
Without relying on any dependence constraint, it has been proved the tail of the cumulative capacity, i.e., $S(t) = C(1) + \cdots + C(t)$, is bounded by \cite{jiang2008stochastic}
\begin{equation}
\overline{F}_{S(t)}(x) \le \overline{F}_{C(1)}\otimes\ldots\otimes\overline{F}_{C(t)}(x).
\end{equation}
In case that the instantaneous capacity is light tailed, i.e., $\overline{F}_{C}(x)\le{a}e^{-bx}$, the tail of the cumulative capacity is also exponentially bounded, i.e.,
\begin{equation}
\overline{F}_{S(t)}(x) \le \prod_{k=1}^{t}(a_{k}b_{k}w)^{\frac{1}{b_{k}w}} \cdot e^{\frac{-x}{w}},
\end{equation}
where $w=\sum_{k=1}^{t}\frac{1}{b_k}$, by applying a distribution bound for the sum of exponentially bounded random variables  \cite{jiang2008stochastic}.

For a concatenation of wireless channels, each with a cumulative capacity $S_i(s,t)$, the cumulative capacity of the concatenation network can be expressed 
\begin{equation}\label{eq-cc}
S(s,t) = S_1\otimes\ldots\otimes S_N(s,t).
\end{equation}
This is because the cumulative capacity process is essentially the service process of the channel. Consequently, the service process SNC result for such a concatenation system directly applies, which has the same form as (\ref{eq-cc}) (see e.g., \cite{fidler2006end} \cite{jiang2008stochastic}).

Then, the tail is expressed as
\begin{IEEEeqnarray}{rCl}
\overline{F}_{S(t)}(x) &=& P\left\{ S_1\otimes\ldots\otimes S_N(t) \ge x \right\} \\
&=& P\left\{ \inf_{\mathbf{u}\in\mathcal{U}(x)}\sum_{i=1}^N {S_i}(\mu_{i-1},\mu_i) \ge x \right\} \IEEEeqnarraynumspace\\
&\le& \inf_{\mathbf{u}\in\mathcal{U}(x)} P\left\{ \sum_{i=1}^N {S_i}(\mu_{i-1},\mu_i) \ge x \right\} \\
&\le& \inf_{\mathbf{u}\in\mathcal{U}(x)} E\left[ e^{\theta\sum_{i=1}^N {S_i}(\mu_{i-1},\mu_i)} \right] \cdot e^{-\theta{x}}, \IEEEeqnarraynumspace
\end{IEEEeqnarray}
where $\mathcal{U}(x)=\{ \mathbf{u}=(u_1,\ldots,u_t): {0\le{\mu_1}\le{\ldots}\le{\mu_{N-1}}\le{t}} \}$, for some $\theta>0$. This ends the proof.
\end{proof}

\section{Distribution Bounds}\label{bounds}

In general, there is dependence in capacity over time, for which the dependence in fading is a direct cause.

For the influence of stochastic dependence in the cumulative capacity process, the general Fr\'{e}chet bounds \cite{ruschendorf2013mathematical} directly apply as:
\begin{equation}
\widecheck{F}_{S(t)}(x) \le F_{S(t)}(x) \le \widehat{F}_{S(t)}(x),
\end{equation}
where
\begin{eqnarray}
\widecheck{F}_{S(t)}(x) &=& \left[ {\sup_{\mathbf{u}\in\mathcal{U}(x)} \sum_{i=1}^{t}F_{C(i)}(u_i)} -(t-1)\right]^{+} \\
\widehat{F}_{S(t)}(x) &=& \left[ {\inf_{\mathbf{u}\in\mathcal{U}(x)} \sum_{i=1}^{t}F_{C(i)}(u_i) } \right]_{1}
\end{eqnarray}
with $\mathcal{U}(x)=\left\{ \mathbf{u}=(u_1,\ldots,u_t):\sum_{i=1}^{t}u_i=x \right\}$. 

The Fr\'{e}chet bounds are general bounds. By making use of the dependence information among $C(1), C(2), \dots$, the bounds may be improved. In addition, the bound analysis of cumulative capacity can be extended to delay performance analysis. To this aim, three specific and representative capacity process types are investigated in this section, which are comonotonic process, additive process, and Markov additive process.

\subsection{Comonotonic Process}

The upper Fr\'{e}chet bound expresses the extremal positive dependence indicating the largest sum with respect to convex order and the dependence structure is represented by the comonotonic copula, i.e., \cite{dhaene2002concept,embrechts2003using,embrechts2006bounds}
\begin{equation}
F(c_{1},\ldots,c_{t}) = \min_{1\le{i}\le{t}} F_{C(i)}(c_i),
\end{equation}
or equivalently \cite{dhaene2002concept}, for $U\sim U(0,1)$,
\begin{equation}
(C(1),\ldots,C(t)) \stackrel{d}{=} \left(F^{-1}_{C(1)}(U),\ldots,F^{-1}_{C(t)}(U)  \right),
\end{equation}
which implicates that comonotonic random variables are increasing functions of a common random variable \cite{embrechts2003using}.

If the increment of the cumulative capacity carries a dependence structure of comonotonicity, we define the cumulative capacity as a comonotonic process in this paper, which is different from a similar concept regarding the comonotonicity between different processes \cite{jouini2003comonotonic}.

\begin{theorem}
For a stationary capacity process, the distribution of the cumulative capacity with comonotonic dependence is expressed as
\begin{equation}
F_{S(t)}(x) = F_{C}\left( \frac{x}{t} \right).
\end{equation}
\end{theorem}

\begin{proof}
In the special case that all marginal distribution functions are identical $F_{C(i)}\sim F_{C}$, comonotonicity of $C(i)$ is equivalent to saying that $C(1)=C(2), \ldots, =C(t)$ holds almost surely \cite{dhaene2002concept}. 
In other words, 
the sample function of the capacity process is stationary and depends only on the initial value of the capacity in each realization.
\end{proof}

\begin{theorem}
Consider a constant arrival process $A(t)=\lambda{t}$. The delay at finite time horizon is expressed as
\begin{equation}
P(D(t) > d) = P\left\{ C(1) < \lambda - \frac{\lambda{d}}{t} \right\},
\end{equation}
while the delay at infinite time horizon is expressed as
\begin{equation}
P(D>d; \forall{d}>0) = P\{ C(1) < \lambda \}.
\end{equation}
\end{theorem}

\begin{proof}
For a constant arrival process $A(t)=\lambda{t}$, the delay is expressed as
\begin{IEEEeqnarray}{rCl}
P(D(t) > d) &=& P\left\{ \sup_{{0}\le{s}\le{t}}({A}(s)-{S}(s)) > {\lambda{d}} \right\} \\
&=& P\left\{ C(1) < \lambda - \frac{\lambda{d}}{t} \right\}.
\end{IEEEeqnarray}
Letting time go to infinity gives
\begin{equation}
P(D>d; \forall{d}>0) = P\{ C(1) < \lambda \}.
\end{equation}
This completes the proof.
\end{proof}

It indicates that, for a comonotonic capacity process, a delay bound makes sense only in the finite time horizon, and for the infinite time horizon, whenever there is deep fade, there will be infinite delay, which is relevant to the outage probability for slow fading \cite{tse2005fundamentals}.


\subsection{Additive Process}

The independence structure of an additive process is expressed by a product copula 
\begin{equation}
F(c_{1},\ldots,c_{t}) = \prod_{i=1}^{t}F_{C(i)}(c_i),
\end{equation} 
and the distribution of the cumulative capacity is expressed via Stieltjes convolution as
\begin{equation}
F_{S(t)}(x) = F_{C(1)}\circledast\ldots\circledast F_{C(t)}(x).
\end{equation}
The cumulative capacity with independent increment can be modeled as an additive process \cite{ito2006essentials}.

\begin{theorem}
For a stationary capacity process, the distribution of the cumulative capacity with independence is expressed as, for some $\theta>0$,
\begin{equation}
1 - e^{t\kappa(\theta)-\theta{x}} \le F_{S(t)}(x) \le e^{t\kappa(-\theta)+\theta{x}},
\end{equation}
where $\kappa(\theta) = \log{E}\left[e^{\theta{C(i)}}\right]$ is the cumulant generating function of the instantaneous capacity,
and the distribution of the transient capacity is expressed as
\begin{equation}
1 - e^{-y_l} \le P\left\{ \overline{C}(t) \le c^\ast \right\}\le e^{-y_u},
\end{equation}
where $c^\ast=\frac{t\kappa(\theta^\ast)+y^\ast}{\theta^\ast{t}}$, with $y^\ast=y_u$ for $\theta^\ast<{0}$ for the upper bound, and $y^\ast=y_l$ for $\theta^\ast>{0}$ for the lower bound.
\end{theorem}

\begin{proof}
In the special case that all marginal distribution functions are identical $F_{C(i)}\sim F_{C}$, a likelihood ratio process of the cumulative capacity can be formulated and is expressed as \cite{asmussen2003applied}
\begin{equation}
L(t) = e^{\theta{S(t)}-t\kappa(\theta)},
\end{equation}
where $L(t)$ is a mean-one martingale and $\kappa(\theta)$ is the cumulant generating function, i.e., 
\begin{eqnarray}
\kappa(\theta) = \log{E}\left[e^{\theta{C(i)}}\right] = \log\int e^{\theta{x}} F(dx),
\end{eqnarray}
where $\theta\in\Theta=\{ \theta\in\mathbb{R}:\kappa(\theta)<\infty \}$.

According to Markov inequality, for any $\mu>0$,
\begin{equation}
P\{ L(t)\ge{\mu} \} \le\frac{1}{\mu}{E}[L(t)]=\frac{1}{\mu}.
\end{equation}
Letting $\mu = e^{-t\kappa(\theta)+\theta{x}}$,
for $\theta\le{0}$, the cumulative distribution function is bounded by
\begin{equation}
P\{ S(t)\le x \} \le e^{t\kappa(\theta)-\theta{x}},
\end{equation}
while for $\theta>0$, the complementary cumulative distribution function is expressed as
\begin{equation}
P\{ S(t)\ge x \} \le e^{t\kappa(\theta)-\theta{x}},
\end{equation}
which shows that the distribution has a light tail.
Letting $-y^\ast = t\kappa(\theta)-\theta{x} \le{0}$, 
the distribution of the transient capacity is bounded by
\begin{equation}
1 - e^{-y_l} \le P\left\{ \overline{C}(t) \le c^\ast \right\}\le e^{-y_u},
\end{equation}
where $c^\ast=\frac{t\kappa(\theta)+y^\ast}{\theta{t}}$, with $y^\ast=y_u$ for $\theta<{0}$ for the upper bound, and $y^\ast=y_l$ for $\theta>{0}$ for the lower bound.
\end{proof}

\begin{theorem}
Consider a constant arrival process $A(t)=\lambda{t}$, the delay at the wireless channel is bounded by
\begin{equation}
C_{-}e^{-\theta{\lambda{d}}}\le P(D\ge{d}) \le C_{+}e^{-\theta{\lambda{d}}}.
\end{equation}
Letting $P(D\ge{d})=\epsilon$, the delay constrained capacity is expressed by
\begin{equation}
\frac{-\log\frac{\epsilon}{C_{-}}}{\theta{d}} \le \lambda \le \frac{-\log\frac{\epsilon}{C_{+}}}{\theta{d}},
\end{equation}
where
\begin{eqnarray}
C_{-} &=& \inf_{x\in[0,x_0)}\frac{\overline{B}(x)}{\int_{x}^{\infty}e^{\theta(y-x)}B(dy)}, \\
C_{+} &=& \sup_{x\in[0,x_0)}\frac{\overline{B}(x)}{\int_{x}^{\infty}e^{\theta(y-x)}B(dy)},
\end{eqnarray}
and $B$ is the distribution of $\lambda-C$ and $x_0=\sup\{x:B(x)<1\}$.
\end{theorem}

\begin{proof}
For a constant arrival process $A(t)=\lambda{t}$, the delay is bounded by
\begin{IEEEeqnarray}{rCl}
P(D \ge d) &=& P\left\{ \sup_{t\ge{0}}({A}(t)-{S}(t)) \ge{\lambda{d}} \right\} \\
&\le& e^{-\theta{\lambda{d}}},
\end{IEEEeqnarray}
where the last inequality follows the Lundberg's inequality \cite{rolski1998stochastic,asmussen2010ruin}, if $\theta(>0)$ satisfies the Lundberg equation $\kappa(\theta)=0$, where
\begin{equation}
\kappa(\theta)= \log\int e^{\theta(\lambda-C(t))} F(dx).
\end{equation}
The approach to obtain the lower bound and to improve the prefactors is available in \cite{rolski1998stochastic,asmussen2010ruin}.
\end{proof}

\subsection{Markov Additive Process}


The Markov property is a pure dependence property that can be formulated exclusively in terms of copulas. As a consequence, starting with a Markov process, a multitude of other Markov processes can be constructed by just modifying the marginal distributions \cite{darsow1992copulas,overbeck2015multivariate}. 

Specifically, if dependence in the capacity follows a Markov process, the instantaneous capacity has a specific distribution with respect to a state transition, then the cumulative capacity is a Markov additive process.

\begin{theorem}
For a Markov additive process, the distribution of the cumulative capacity is expressed as, for some $\theta>0$,
\begin{equation}
1 - \frac{h^{(\theta)}(J_0) e^{t\kappa(\theta)-\theta{x}} }{\min\limits_{j\in{E}}(h^{(\theta)}(J_j))} \le F_{S(t)}(x) \le \frac{h^{(-\theta)}(J_0) e^{t\kappa(-\theta)+\theta{x}} }{\min\limits_{j\in{E}}(h^{(-\theta)}(J_j))},
\end{equation}
and the distribution of the transient capacity is expressed as
\begin{equation}
1 - \frac{h^{(\theta)}(J_0){e^{-y_l}}}{\min\limits_{j\in{E}}(h^{(\theta)}(J_j))} \le P\left\{ \overline{C}(t) \le c^\ast \right\}\le \frac{h^{(\theta)}(J_0){e^{-y_u}}}{\min\limits_{j\in{E}}(h^{(\theta)}(J_j))},
\end{equation}
where $c^\ast=\frac{t\kappa(\theta^\ast)+y^\ast}{\theta^\ast{t}}$, with $y^\ast=y_u$ for $\theta^\ast<{0}$ for the upper bound, and $y^\ast=y_l$ for $\theta^\ast>{0}$ for the lower bound.
\end{theorem}

A Markov additive process is defined as a bivariate Markov process $\{X_t\}=\{(J_t,S(t))\}$ where $\{J_t\}$ is a Markov process with state space $E$ and the increments of $\{S(t)\}$ are governed by $\{J_t\}$ in the sense that \cite{asmussen2010ruin}
\begin{equation}
{E}[f(S({t+s})-S(t))g(J_{t+s})|\mathscr{F}_t] = {E}_{J_t,0}[f(S(s))g(J_s)].
\end{equation}
For finite state space and discrete time, a Markov additive process is specified by the measure-valued matrix (kernel) $\mathbf{F}(dx)$ whose $ij$th element is the defective probability distribution 
\begin{equation}
F_{ij}(dx)={P}_{i,0}(J_1=j,Y_1\in{dx}),
\end{equation}
where $Y_t=S(t)-S(t-1)$. An alternative description is in terms of the transition matrix $\mathbf{P}=(p_{ij})_{i,j\in{E}}$ (here $p_{ij}={P}_i(J_1=j)$) and the probability measures
\begin{equation}
H_{ij}(dx) = {P}(Y_1\in{dx}|J_0=i, J_1=j) = \frac{F_{ij}(dx)}{p_{ij}}.
\end{equation} 
Consider the matrix $\widehat{\textbf{F}}_t[\theta]=(E_i[e^{\theta{S(t)}};J_t=j])_{i,j\in{E}}$, it is proved that \cite{asmussen2003applied}
\begin{equation}
\widehat{\textbf{F}}_t[\theta]=\widehat{\textbf{F}}[\theta]^t,
\end{equation}
where $\widehat{\textbf{F}}[\theta]=\widehat{\textbf{F}}_1[\theta]$ is a $E\times{E}$ matrix with $ij$th element $\widehat{F}^{(ij)}[\theta]=p_{ij}\int{e^{\theta{x}}F^{(ij)}}(dx)$, and $\theta\in\Theta=\{ \theta\in\mathbb{R}:\int{e^{\theta{x}}F^{(ij)}}(dx)<\infty \}$. By Perron-Frobenius theory, $e^{\kappa(\theta)}$ and $\textbf{h}^{(\theta)}=(h_{i}^{(\theta)})_{i\in{E}}$ are respectively the positive real eigenvalue with maximal absolute value and the corresponding right eigenvector of $\widehat{\textbf{F}}[\theta]$, i.e., $\widehat{\textbf{F}}[\theta]\textbf{h}^{(\theta)}=e^{\kappa(\theta)}\textbf{h}^{(\theta)}$.
In addition, for the left eigenvector $\textbf{v}^{(\theta)}$, $\textbf{v}^{(\theta)}\textbf{h}^{(\theta)}=1$ and $\mathbf{\pi}\textbf{h}^{(\theta)}=1$, where $\mathbf{\pi}=\textbf{v}^{(0)}$ is the stationary distribution and $\textbf{h}^{(0)}=\textbf{e}$. 

\begin{proof}
Like the independent case, a likelihood ratio process can be formulated with an exponential change of measure \cite{asmussen2003applied}:
\begin{equation}
L(t) = \frac{h^{(\theta)}(J_t)}{h^{(\theta)}(J_0)}e^{\theta{S(t)}-t\kappa(\theta)},
\end{equation}
which is a mean-one martingale. 
In order to provide exponential upper bound for the distribution of the cumulative capacity, define \cite{gallager2013stochastic}
\begin{equation}
\underline{L}(t) = \frac{\min_{j\in{E}}(h^{(\theta)}(J_j))}{h^{(\theta)}(J_0)}e^{\theta{S(t)}-t\kappa(\theta)},
\end{equation}
where $\underline{L}(t)\le L(t)$, i.e., ${E}[\underline{L}(t)]\le{1}$. 
Apply Markov inequality to $\underline{L}(t)$ and get, for any $\mu>0$, 
\begin{eqnarray}
P\{ \underline{L}(t)\ge{\mu} \} \le \frac{1}{\mu}{E}[\underline{L}(t)]\le \frac{1}{\mu}.
\end{eqnarray}
Choose $\mu=e^{-t\kappa(\theta)+\theta{x}}\frac{\min_{j\in{E}}(h^{(\theta)}(J_j))}{h^{(\theta)}(J_0)}$, for $\theta\le{0}$, 
\begin{eqnarray}
P\{ S(t)\le x \} \le \frac{h^{(\theta)}(J_0)}{\min_{j\in{E}}(h^{(\theta)}(J_j))}e^{t\kappa(\theta)-\theta{x}},
\end{eqnarray}
while for $\theta>0$,
\begin{eqnarray}
P\{ S(t)\ge x \} \le \frac{h^{(\theta)}(J_0)}{\min_{j\in{E}}(h^{(\theta)}(J_j))}e^{t\kappa(\theta)-\theta{x}},
\end{eqnarray}
which indicates that the distribution has a light tail.
Letting $-y^\ast = t\kappa(\theta)-\theta{x} \le{0}$, 
the distribution of the transient capacity is bounded by
\begin{equation}
1 - \frac{h^{(\theta)}(J_0){e^{-y_l}}}{\min\limits_{j\in{E}}(h^{(\theta)}(J_j))} \le P\left\{ \overline{C}(t) \le c^\ast \right\}\le \frac{h^{(\theta)}(J_0){e^{-y_u}}}{\min\limits_{j\in{E}}(h^{(\theta)}(J_j))},
\end{equation}
where $c^\ast=\frac{t\kappa(\theta)+y^\ast}{\theta{t}}$, with $y^\ast=y_u$ for $\theta<{0}$ for the upper bound, and $y^\ast=y_l$ for $\theta>{0}$ for the lower bound.
\end{proof}

\begin{theorem}
Consider a constant arrival process $A(t)=\lambda{t}$, the delay conditional on the initial state $J_0=i$ is bounded by
\begin{equation}
\frac{h^{(-\theta)}(J_i) e^{-\theta{\lambda{d}}} }{\max\limits_{j\in{E}}h^{(-\theta)}(J_j)} \le P_i(D\ge{d}) \le \frac{h^{(-\theta)}(J_i) e^{-\theta{\lambda{d}}} }{\min\limits_{j\in{E}}h^{(-\theta)}(J_j)},
\end{equation}
and the stationary delay is thus bounded by 
\begin{equation}
P(D\ge d) = \sum_{i\in{E}}\pi_{i}P_i(D_M\ge d).
\end{equation}
Letting $P(D\ge d)=\epsilon$, the delay-constrained capacity is expressed as
\begin{equation}
\frac{-1}{\theta{d}}{\log\frac{\epsilon\cdot{\max\limits_{j\in{E}}h^{(-\theta)}(J_j)}}{\sum_{i\in{E}}{\pi_i}h^{(-\theta)}(J_i)}} \le \lambda \le \frac{-1}{\theta{d}}{\log\frac{\epsilon\cdot{\min\limits_{j\in{E}}h^{(-\theta)}(J_j)}}{\sum_{i\in{E}}{\pi_i}h^{(-\theta)}(J_i)}}.
\end{equation}
\end{theorem}

\begin{proof}
For a constant arrival process $A(t)=\lambda{t}$, the delay conditional on initial state $J_0=i$ is bounded by \cite{zhu2008ruin}
\begin{IEEEeqnarray}{rCl}
P_i(D\ge{d}) &=& P_i\left\{ \sup_{t\ge{0}}(\lambda(t)-{S}(t)) \ge{\lambda{d}} \right\} \IEEEeqnarraynumspace\\
&\le& \frac{h^{(-\theta)}(J_i)}{\min_{j\in{E}}h^{(-\theta)}(J_j)}e^{-\theta{\lambda{d}}},
\end{IEEEeqnarray}
where the last inequality follows the Lundberg's inequality, if $\theta(>0)$ satisfies the Lundberg equation $\kappa(-\theta)=0$. $\kappa(\theta)$ is the logarithm of the Perron-Frobenius eigenvalue of the kernel for the Markov additive process $C(t)-\lambda$, i.e., $\widehat{\textbf{F}}[\theta]$. 
The lower delay bound is available in \cite{zhu2008ruin}.
\end{proof}

Specifically, if $F_{ij}$ is independent of $j$, the prefector in the Lundberg inequality can be improved and the doubly-sided bound is expressed as, i.e.,
\begin{equation}
C_{-}{h^{(-\theta)}(J_i)}e^{-\theta{\lambda{d}}}\le P_i(D\ge d) \le C_{+}{h^{(-\theta)}(J_i)}e^{-\theta{\lambda{d}}},
\end{equation}
where
\begin{eqnarray}
C_{-} &=& \min_{j\in{E}}\frac{1}{h_j^{(-\theta)}}\cdot\inf_{x\ge{0}}\frac{\overline{B}_j(x)}{\int_{x}^{\infty}e^{\theta(y-x)}B_j(dy)}, \\
C_{+} &=& \max_{j\in{E}}\frac{1}{h_j^{(-\theta)}}\cdot\sup_{x\ge{0}}\frac{\overline{B}_j(x)}{\int_{x}^{\infty}e^{\theta(y-x)}B_j(dy)},
\end{eqnarray}
and $B_j$ is the distribution of the instantaneous capacity $C_j$ \cite{asmussen2010ruin}.

\begin{remark}
The Markov additive process can be seen as a non-stationary additive process defined on a Markov process. If the Markov process has only one state, then it reduces to a stationary additive process \cite{ccinlar1973theory}. 
\end{remark}

\section{Channel Comparisons}\label{comparisons}

\subsection{Ordering of Cumulative Capacity}

The cumulative capacity $S_X$ is said to be smaller than $S_Y$ in stochastic order, written as
\begin{equation}
S_X\le_{st}S_Y, 
\end{equation}
if the distribution functions $F_{S_X}$ and $F_{S_Y}$ are comparable in the sense that $P(S_X\le{x}) \ge P(S_Y\le{x})$, $\forall{x}$.
An equivalent condition is that the expectation of all increasing functions $\mathcal{F}$ is larger for $S_Y$ than for $S_X$, i.e., $E[f(S_X)] \le E[f(S_Y)],~\forall{f}\in\mathcal{F}$.

The cumulative capacity $S_X$ is said to be smaller than $S_Y$ in convex order (respectively increasing convex order), written as
\begin{equation}
S_X\le_{cx}S_Y, 
\end{equation}
(respectively $S_X\le_{icx}S_Y$), 
if for all convex functions $\mathcal{F}_{cx}$ (respectively all increasing convex functions $\mathcal{F}_{icx}$, $E[f(S_X)]\le E[f(S_Y)]$, $\forall{f}\in\mathcal{F}_{cx}$ (respectively $\forall{f}\in\mathcal{F}_{icx}$).

By comparing to the probability measure of independence, positive dependence and negative dependence can be defined under specific stochastic orderings.
In particular, the cumulative capacity $S$ is said to have a positive dependence structure in the sense of increasing convex order, if
\begin{equation}
S_{\perp} \le_{icx} S_P,
\end{equation}
and respectively a negative dependence structure in the sense of increasing convex order, if
\begin{equation}
S_N \le_{icx}S_{\perp},
\end{equation}
where $S_{\perp}$ has an independence structure. 
Positive dependence implies that large values of random variables tend to occur together, while negative dependence implies that large values of one variable tend to occur together with small values of others \cite{denuit2006actuarial}.

\begin{lemma}
For cumulative capacities $S_N$, $S_{\perp}$ and $S_P$ which respectively have negative dependence, independence, and positive dependence structures, if their marginal distributions are identical for all $t$, their convex ordering is equivalent to their increasing convex ordering, i.e.,
\begin{multline}
S_N \le_{icx} S_{\perp} \le_{icx}S_P \Longleftrightarrow
S_N \le_{cx} S_{\perp} \le_{cx}S_P.
\end{multline}
\end{lemma}

\begin{proof}
Since the mean of the sum of random variables equals the sum of the means of individual random variables, i.e.,
\begin{equation}
E[S_N]= E[S_{\perp}]= E[S_P],
\end{equation}
the proof follows directly from that the increasing convex order is identical to the convex order under equal expectations \cite{kampke2015income}.
\end{proof}

\subsection{Ordering of Delay}

For dependence scenarios with more involved assumptions, explicit results, e.g., for delay, are hard to derive or no more tractable. Yet, if their dependence structures are known, insights are obtained if delay ordering may be found from the dependence structure, which is the focus of this subsection.  

By making use of Chernoff bounds, an upper bound on delay is as: for some $\theta>0$, 
\begin{eqnarray}
P(D\ge{d}) &\le& \sum_{t=0}^{\infty} P\left\{ { A(t) - S(t) }\ge\lambda{d} \right\} \\
&\le& e^{-\theta\lambda{d}}\sum_{t=0}^{\infty}E\left[e^{\theta(\lambda{t}-S(t))}\right],
\end{eqnarray}
which indicates that the complementary delay distribution has an exponential bound with adjustment coefficient $\theta$, if the instantaneous capacity is light-tailed.
Particularly, for light tailed instantaneous capacity, the asymptotic behavior of the bounding function is still exponential for week forms of dependence, while it becomes heavy tailed for stronger dependence \cite{asmussen2010ruin}. 

\begin{theorem}
Consider two wireless channel capacity processes, if the cumulative capacities are convex ordered, then the above adjustment coefficients for the delay bounds are correspondingly ordered, i.e.,
\begin{equation}
S\le_{cx}\widetilde{S}
\Longrightarrow
\widetilde{\theta} \ge \theta.
\end{equation}
\end{theorem}

\begin{proof}
Consider the negative increment process, i.e.,
\begin{equation}
-X(t)=C(t)-a(t). 
\end{equation}
If it is light-tailed, then the delay violation probability has an exponential bound with adjustment coefficient $\theta$ defined by $\kappa(\theta)=0$, where \cite{asmussen2010ruin,muller2001asymptotic}
\begin{equation}
\kappa(\theta) = \lim_{t\rightarrow\infty}\frac{1}{t}E\left[ e^{\theta\sum_{i=1}^{t}{X(i)}} \right].
\end{equation}
By exploring the ordering of the cumulative increment process, 
\begin{equation}
\sum_{i=1}^{n}-X(i) \le_{cx} \sum_{i=1}^{n}-\widetilde{X}(i),
\end{equation}
the adjustment coefficients are ordered as follows \cite{asmussen2010ruin,muller2001asymptotic}
\begin{equation}
\widetilde{\theta} \ge \theta.
\end{equation}
Specifically, for constant arrival, the ordering of the cumulative capacity results in the ordering of the cumulative negative increment process. 
\end{proof}

Since every multi-dimensional distribution functions $\textbf{y}\mapsto I(\textbf{y}\le\textbf{x})$ and multi-dimensional survival functions $\textbf{y}\mapsto I(\textbf{y}>\textbf{x})$ are both supermodular functions \cite{shaked2007stochastic}, i.e.,
$f(\textbf{x}) + f(\textbf{y}) \le f(\textbf{x}\wedge\textbf{y}) + f(\textbf{x}\vee\textbf{y})$,
the supermodular ordering of the instantaneous increment, i.e.,
\begin{equation}
-\textbf{X} \le_{sm} -\widetilde{\textbf{X}},
\end{equation}
indicates that the marginal distributions of the instantaneous increments are identical, 
which can be used for comparison between scenarios of instantaneous increment with identical marginal distributions and different dependence structures. Specifically, if $-\textbf{X}\le_{sm}-\widetilde{\textbf{X}}$, then $\sum_{i=1}^{n}-X(i)\le_{cx}\sum_{i=1}^{n}-\widetilde{X}(i)$.

\begin{corollary}
For delay bounding functions with the same prefactor or are bounded by a same prefactor before the exponential term, the ordering of the cumulative capacity $S_N \le_{cx} S_{\perp} \le_{cx}S_P$ indicates the ordering of the delay, i.e.,
\begin{equation}
P(D_N\ge{x}) \le P(D_{\perp}\ge{x}) \le P(D_P\ge{x}),
\end{equation}
and the ordering of the delay-constrained capacity for the same prefactor, i.e.,
\begin{equation}
\lambda_N \ge \lambda_{\perp} \ge \lambda_P.
\end{equation}
 \end{corollary}

Based on the ordering result, to analyze the delay violation probability of an intractable dependence scenario, we may use that of a scenario that has some special dependence structure, whose analysis is tractable, to bound. For instance, the results  under independence assumption can be treated as a conservative approximation for cases with negative dependence structures. 

\section{Interferences in Channel}\label{interference}

For an ad hoc network, nodes may communicate in a multi-hop style, where, the output from the previous hop is exactly the input to the next hop \cite{li2001capacity}.
A simple scenario for self-interference is neighbor interference, i.e., interference only exists in adjacent hops. In this case, the total input to the channel consists of both the output from the previous hop and the output of the channel which is input to the next hop. Overall, the wireless channel can be treated as a feedback system.

\subsection{Single Hop Case}
We focus on the delay upper bound. 

Consider a wireless channel with input $A(t)$ and output $A^\ast(t)$, where the output $A^\ast(t)$ is directly fed back into the wireless channel, i.e., the total input $\widetilde{A}(t)$ to the channel is 
\begin{equation}
\widetilde{A}(t) = A(t)+A^\ast(t). 
\end{equation}

For such a feedback system, we can treat it as a blackbox providing service $\widetilde{S}(t) $ only to the input $A(t)$. Then, as shown in Eq. (\ref{eq-ior})
\begin{equation}
A^{\ast}(t) = A\otimes \widetilde{S}(t). 
\end{equation}
In the more general case, if the output $A^\ast(t)$ first passes through a server with capacity process $\mathring{S}(t)$ on the path of feedback, then, the overall input to the channel becomes
\begin{equation}
\widetilde{A}(t) = A(t) + A^\ast\otimes \mathring{S}(t).
\end{equation}

The following theorem establishes a relation between $\widetilde{S}(t)$, $S(t)$ and $A(t)$.

\begin{theorem}
The service process $\widetilde{S}(t)$ for the input $A(t)$ is lower bounded by
\begin{equation}
\widetilde{S}(t) \ge S(t) - A(t),
\end{equation}
correspondingly, the delay is upper bounded by
\begin{equation}
P(D\ge d) 
\le P\left\{ \sup_{t\ge{0}}( A(t) + A(t) - {S}(t) ) \ge A(d) \right\}.
\end{equation}
\end{theorem}

\begin{proof}
The service for the input $A(t)$ is bounded by 
\begin{eqnarray}
\widetilde{S}(t) &\ge& S(t) - A^\ast\otimes \mathring{S}(t) \\
&\ge& S(t) - A^\ast(t) \\
&\ge& S(t) - A(t),
\end{eqnarray}
where the first inequality follows the leftover service under blind scheduling \cite{jiang2008stochastic}, the second inequality follows the monotonicity of bivariate min-plus convolution \cite{chang2000performance}, i.e., $\forall{t}$, $f\otimes{g}\le{g}$ if $f(t,t)=0$ or $f\otimes{g}\le{f}$ if $g(t,t)=0$, and the last inequality takes advantage of system causality, i.e., $A(t)\ge A^{\ast}(t)$.

Then the delay is bounded by
\begin{equation}
P(D\ge d) 
\le P\left\{ \sup_{t\ge{0}}(A(t) - (S(t) - A(t)) \ge A(d) \right\}
\end{equation}
which is resulted from a direct application of the SNC delay bound analysis, given the arrival process $A(t)$ and the service process $\widetilde{S}(t) \ge S(t) - A(t)$ \cite{jiang2008stochastic}.
\end{proof}

\begin{example}[Additive Case]
For the constant arrival process $A(t)=\lambda{t}$, the delay is bounded by
\begin{IEEEeqnarray}{rCl}
P(D\ge d) 
&\le& P\left\{ \sup_{t\ge{0}}(2\lambda(t)-{S}(t)) \ge{d\lambda} \right\} \IEEEeqnarraynumspace\\
&\le& e^{-\theta{d\lambda}},
\end{IEEEeqnarray}
where the last inequality follows Lundberg's inequality, if $\theta>0$ satisfies the Lundberg equation $\kappa(\theta)=0$, where
$\kappa(\theta)= \log\int e^{\theta(2\lambda-C(t))} F(dx)$.
\end{example}

\begin{example}[Markov Additive Case]
For the constant arrival process $A(t)=\lambda{t}$, the delay conditional on initial state $J_0=i$ is bounded by
\begin{IEEEeqnarray}{rCl}
P_i(D\ge d) 
&\le& P_i\left\{ \sup_{t\ge{0}}(2\lambda(t)-{S}(t)) \ge{d\lambda} \right\} \IEEEeqnarraynumspace\\
&\le& \frac{h^{(-\theta)}(J_i)}{\min_{j\in{E}}h^{(-\theta)}(J_j)}e^{-\theta{d\lambda}},
\end{IEEEeqnarray}
where the last inequality follows Lundberg's inequality, if $\theta>0$ satisfies the Lundberg equation $\kappa(-\theta)=0$.
$\kappa(\theta)$ is the logarithm of the Perron-Frobenius eigenvalue of the kernel for the Markov additive process $C(t)-2\lambda$, i.e., $\widehat{\textbf{F}}[\theta]$. 
Then the delay is bounded by $P(D\ge d) \le \sum_{i\in{E}}\pi_{i}P_i(D\ge d)$. 
\end{example}

\subsection{Multiple Hop Case}

For a concatenation of wireless channels, assume a common wireless channel is shared among different hops, i.e., the arrival goes through a common channel for multiple times while each time it sees a different channel capacity, thus the end to end capacity is expressed as
\begin{IEEEeqnarray}{rCl}
\IEEEeqnarraymulticol{3}{l}{
(S-A^{\ast}_1)\otimes\ldots\otimes(S-A^{\ast}_N)(t)
} \\\qquad
&\ge& (S-A_1)\otimes\ldots\otimes(S-A_1)(t) \\
&=& \inf_{\mathbf{u}\in\mathcal{U}(x)}\sum_i^N (S-A_1)(\mu_{i-1},\mu_i) \\
&\ge& (S-A_1)(t), 
\end{IEEEeqnarray} 
where $\mathcal{U}(x)=\{ \mathbf{u}=(u_1,\ldots,u_t): {0\le{\mu_1}\le{\ldots}\le{\mu_{N-1}}\le{t}} \}$, $A_1(t)=A(t)$, the first inequality holds because of the monotonicity of the bivariate min-plus convolution \cite{chang2000performance}, i.e., $f\otimes{g}(s,t)\le \tilde{f}\otimes\tilde{g}$, $\forall$ $f\le\tilde{f}$ and $g\le\tilde{g}$, and the second inequality holds under the assumption that $S(t)-A(t)$ is a subadditive process \cite{kingman1976subadditive}, e.g., a stationary additive process. 
It implies that, for neighbor interference, the number of hops has no impact on the end to end service of the feedback system, i.e., it seems that the arrival traverses the channel only once.

The neighbor interference is the extremal scenario that only output interference should be considered. For the generic $K$ hop interference, where $K$ is independent of the network size $N$ in principle, both output and input interference should be taken into account and the most severe interference contains $K$ output interference and $K-1$ input interference. 
In contrast to output interference, input interference is the interference to the next hops.
Under the same assumption for neighbor interference and with the same approach for analysis, the network service is lower bounded by
\begin{eqnarray}
S(t)-(2K-1)A(t),
\end{eqnarray}
where $K=\min(K,N)$.
It is worth noting that the interference of the input is absolute while the interference of the output is relative in that it exists only when the output is fed back into the wireless channel.
This result implies that the network service is limited by the hop with the most severe interference, and the multi-hop network can be reduced to a single hop system for performance analysis.

The above insights are summarized in the following corollary.

\begin{corollary}
Consider a concatenation of wireless channels with intra- and inter- interferences. If all hops see a common wireless channel, i.e., an identical service process is shared among each hop, then the number of hops has no impact on performance, i.e., it seems that the arrival traverses through the channel for only once with a service process that is limited by the hop experiencing the deepest interference.
\end{corollary}

Another generalization is that different hops see different wireless channels, and the result is summarized in the following theorem. 

\begin{theorem}
Consider a concatenation of wireless channels with cumulative capacity process $S_i(t)$ that is an additive process. Then, for constant arrival $A(t) = \lambda{t}$, the end to end delay is expressed as
\begin{equation}
P(D\ge{d}) 
\le \sum_{t=0}^{\infty}\sum_{\mathbf{u}\in\mathcal{U}(x)} E\left[ e^{-\theta\sum_{i=1}^N {S_i^\ast}(\mu_{i-1},\mu_i)} \right] \cdot e^{\theta{\lambda{(t-d)}}}, 
\end{equation}
where ${S_i^\ast}(\mu_{i-1},\mu_i)=(S_i-K^{\ast}A_1)(\mu_{i-1},\mu_i)$ and $\mathcal{U}(x)=\{ \mathbf{u}=(u_1,\ldots,u_t): {0\le{\mu_1}\le{\ldots}\le{\mu_{N-1}}\le{t}} \}$.
\end{theorem}

\begin{proof}
Recall that the distribution function of the cumulative capacity of a concatenation of wireless channels is bounded by
\begin{IEEEeqnarray}{rCl}
\IEEEeqnarraymulticol{3}{l}{
{F}_{S_t}(x) = P\left\{ S_1\otimes\ldots\otimes S_N(t) \le x \right\} 
} \\\quad
\le \sum_{0\le{\mu_1}\le{\ldots}\le{\mu_{N-1}}\le{t}} E\left[ e^{-\theta\sum_{i=1}^N {S_i}(\mu_{i-1},\mu_i)} \right] \cdot e^{\theta{x}}, \IEEEeqnarraynumspace
\end{IEEEeqnarray}
for some $\theta>0$.

Specifically, the network capacity with interference is bounded by
\begin{IEEEeqnarray}{rCl}
\IEEEeqnarraymulticol{3}{l}{
S(t) =
(S_1 - A^{\ast}_1)\otimes\ldots\otimes(S_N - A^{\ast}_N)(t)
} \\\quad
&\ge& (S_1 - K^{\ast}A_1)\otimes\ldots\otimes(S_N - K^{\ast}A_1)(t),
\end{IEEEeqnarray} 
where $K^\ast=2K-1$ and $K=\min(K,N)$.
Thus the end to end delay is bounded by 
\begin{IEEEeqnarray}{rCl}
\IEEEeqnarraymulticol{3}{l}{
P(D\ge{d}) \le \sum_{t=0}^{\infty} P\left\{ { S(t) }\le\lambda{(t-d)} \right\}
} \\\quad
\le \sum_{t=0}^{\infty}\sum_{\mathbf{u}\in\mathcal{U}(x)} E\left[ e^{-\theta\sum_{i=1}^N {S_i^\ast}(\mu_{i-1},\mu_i)} \right] \cdot e^{\theta{\lambda{(t-d)}}}, \IEEEeqnarraynumspace
\end{IEEEeqnarray}
where ${S_i^\ast}(\mu_{i-1},\mu_i)=(S_i-K^{\ast}A_1)(\mu_{i-1},\mu_i)$ and $\mathcal{U}(x)=\{ \mathbf{u}=(u_1,\ldots,u_t): {0\le{\mu_1}\le{\ldots}\le{\mu_{N-1}}\le{t}} \}$, so long as the summation converges for some $\theta>0$.
\end{proof}

\section{Discussion and Related Work}\label{related}

Network calculus is a promising theory for service guarantee analysis of queueing systems in computer networks, which compliments the classic queueing theory. Its stochastic branch, i.e., stochastic network calculus (SNC), is more relevant to wireless networks owing to the stochastic nature of wireless channel which can hardly be dealt by the deterministic branch of network calculus (DNC). In the SNC literature, a lot of results have been reported (see \cite{fidler2015guide} for a comprehensive survey). Representative works include \cite{jiang2006basic} and \cite{ciucu2012perspectives}. The 2006 paper \cite{jiang2006basic} provides a holistic overview on the basic properties that are needed for the theory. In the 2012 paper \cite{ciucu2012perspectives},  new insights in SNC are provided. In brief, these two papers are more about the fundamental SNC theory itself, which is just a framework of theory. For the application of the SNC theory, more details and characteristics of the considered scenario should be taken into careful account. 

In this paper, we focus on wireless networks, taking over the future work of \cite{jiang2006basic}. The intention is to derive results that are helpful for extending SNC to construct a calculus for wireless networks. We focus on exploring special properties of wireless channel capacity and link them to the foundation of SNC through the cumulative capacity process that is a direct match with the most fundamental service process model in SNC. In addition to the large set of results reported in the paper, there are two findings worth highlighting, which bridge two critical gaps when applying SNC to wireless networks. 

One gap is that, in most existing SNC results, their analysis relies on the assumption that the distribution of the service process is exponential (e.g., \cite{lee1995performance}) or that the moment generating function of the service process exists with finite moments (e.g., \cite{fidler2006end}). However, for wireless networks, no justification of this assumption has been given, weakening the application of SNC to wireless networks. An appealing finding of this paper is that for typical wireless channel models, including Rayleigh, Rice, Nakagami-$m$, Weibull, and lognormal fading channel, their instantaneous capacity and cumulative capacity are both light-tailed, which implies that for such a wireless channel, its service process has moment generating function existent with finite moments. This bridges the first gap. In addition, a direct implication of this finding is that the QoS performance in terms of delay and backlog of the channel is up-bounded by exponential distribution, if the traffic arrival is also light-tailed. 

Another gap is caused by the fact that wireless channel is broadcast channel. As such, when the receiver on a channel forwards its received data traffic in the network case, this forwarded traffic flow competes with the initial flow from the sender, causing self- or intra-flow interference on the channel. This problem belongs to the feedback problem, which is an open problem in SNC \cite{jiang2008stochastic}. In this paper, original results are derived to bridge this gap. These results, not limited to the wireless channel channel case, complement the five basic properties of SNC \cite{jiang2006basic}, thus providing a support to extend the application domain of SNC from feedforward networks to non-feedforward networks.


In comparison with literature SNC based performance analysis of wireless networks, this paper works directly on the stochastic channel capacity process. In addition, the literature investigation mainly focuses on the stochastic service curve characterization of the cumulative service process and on applying it to QoS performance analysis. Little has directly focused on the probability distribution function characteristics of the cumulative capacity. One benefit of the analytical approach in this paper is that it can take the channel capacity characteristics more directly into account, revealing findings and obtaining results that could otherwise not be achieved.   

In the literature, while many SNC-based studies have been conducted for performance analysis of wireless networks, they typically map the wireless channel physical layer models to some link layer modes, e.g., Gilbert-Elliott model \cite{jiang2005analysis,fidler2006wlc15}, finite state Markov channel \cite{zheng2013stochastic,lei2016stochastic}, and other upper layer abstractions \cite{ciucu2013towards}. However, such a mapping form physical layer models to upper layer modelers prevents the analysis from directly exploring special characteristics of the physical wireless channel. In addition, it also hides the ultimate capacity limit, which the wireless channel can achieve, from the analysis, which consequently leads to a conservative estimation of the delay-constrained capacity based on the results obtained. 

A  closely related work with using the wireless channel's physical layer fading behavior as its starting point of analysis is \cite{al2016network}. However, the focus of  \cite{al2016network} is on applying Mellin transforms to the fading process, based on which backlog and delay results are derived in the transform domain. A related work on feedback analysis is \cite{shekaramiz2016network}, where a window flow control system is analyzed for additive service process, but its studied feedback type is different from the self-interference resultant feedback type investigated in this paper.


The idea of taking advantage of specific dependence structures in analysis can also be found in the stochastic network calculus literature, e.g., independent increments \cite{fidler2006end,jiang2008stochastic,jiang2010note} and Markov property \cite{zheng2013stochastic, ciucu2014sharp,poloczek2015service, lei2016stochastic}. However, such diverse dependence structures are investigated separately, without comparing their impacts on the channel capacities. 
In \cite{dong2015copula}, the concept of copula is brought into stochastic network calculus, but the focus is on the arrival process.  
Different from \cite{dong2015copula}, our focus is the cumulative capacity process that essentially is the cumulative service process, while not the arrival process. In addition, besides additive processes or L\'{e}vy  processes and Markov processes, whose copulas are focused and their properties are made use of in the related literature works,  we also consider comonotonicity as a dependence structure when there is a strong time dependence in the channel, i.e., when the time series of instantaneous wireless channel capacities can be represented as increasing functions of a common random variable. 


\section{Conclusion}\label{conclusion}

Future wireless communication calls for exploration of more efficient use of wireless channel capacity to meet the increasing demand on higher data rate and less latency. This motivates to maximally take into consideration the special characteristics of the wireless channel capacity process in analysis, which include the tail behavior, distribution bounds, stochastic ordering, and self-interference in communication of the capacity. To this aim, a set of new results directly exploring these characteristics have been presented in this paper. Among them, an appealing finding is that, for typical fading channels, their instantaneous capacity and cumulative capacity are both light-tailed. It immediately implicates that the cumulative capacity and subsequently the delay and backlog performance can be upper-bounded by some exponential distributions and provides evident justification for the exponential distribution assumption used in the literature. In addition, various bounds have been derived for distributions of the cumulative capacity and the delay-constrained capacity, considering three representative dependence structures in the capacity process, namely comonotonicity, independence, and Markovian. To help gain insights in the performance of a wireless channel, stochastic orders are introduced to the cumulative capacity process, based on which, results to compare the delay and delay-constrained capacity performance have been obtained. Moreover, the open SNC problem, i.e., the impact on performance caused by self-interference in communication in wireless channel is tackled through a novel approach that models it as a feedback system, based on which original results have been derived. In all, the set of results obtained in this paper provide fundamental contributions to linking the SNC theory to wireless networks and hence contribute significantly to its extension towards a calculus for wireless networks. 


\cleardoublepage

\bibliographystyle{abbrv}
\bibliography{main}

\end{document}